\documentclass[a4paper]{article}

\usepackage[utf8]{inputenc}
\usepackage{mathtools}
\usepackage{amsfonts}
\usepackage{amsthm}
\usepackage{amsmath}
\usepackage{amssymb}
\usepackage{enumitem} 
\usepackage{graphicx}
\usepackage{subcaption}
\usepackage{mathrsfs}
\usepackage{a4wide}
\usepackage[text,roman]{esdiff}

\usepackage{tikz}
\usetikzlibrary{decorations.markings}
\tikzstyle{vertex}=[circle, draw, inner sep=0pt, minimum size=3pt, fill]

\usepackage{thmtools}

\theoremstyle{plain}
\newtheorem{theorem}{Theorem}[section]
\newtheorem{proposition}[theorem]{Proposition}
\newtheorem{lemma}[theorem]{Lemma}
\newtheorem{corollary}[theorem]{Corollary}
\theoremstyle{definition}
\newtheorem{definition}[theorem]{Definition}

\numberwithin{equation}{section}

\renewcommand{\l}{\ell}

\newcommand{\N}{\mathbb{N}}
\newcommand{\Q}{\mathbb{Q}}
\newcommand{\R}{\mathbb{R}}
\newcommand{\C}{\mathbb{C}}
\newcommand{\DD}{\mathcal{D}}

\renewcommand{\H}{\mathcal{H}}

\newcommand{\U}{\mathcal{U}}

\renewcommand{\d}{\mathrm{d}}

\DeclareMathOperator{\dom}{dom}

\DeclareMathOperator{\diag}{diag}

\newcommand{\abs}[1]{\left| #1 \right|}
\newcommand{\norm}[1]{\left\Vert #1 \right\Vert}
\newcommand{\scalar}[2]{\langle #1, #2 \rangle}

\newcommand{\carrow}[1]{\overset{\curvearrowright}{#1}}
\newcommand{\veval}[1]{\underline{#1}}
\newcommand{\oveval}[1]{\carrow{\underline{#1}}}
\def\D{\diff}

\newcommand{\Dt}[1][]{\D[#1]{}{t}}
\newcommand{\Dx}[1][]{\D[#1]{}{x}}

\title{Quantum Control at the Boundary}
\author{A. Balmaseda and J.M. Pérez-Pardo}
\date{\today}

\begin{document}

\maketitle

\abstract{%
We introduce a scheme for controlling the state of a quantum system by manipulating its boundary conditions. This contrasts with the usual approach based on direct interactions with the system, that is, by adding interaction terms to the Hamiltonian of the system. We address this infinite-dimensional control problem, providing conditions to the existence of dynamics and approximate controllability for a family of quantum one-dimensional systems.
}


\section{Introduction}
The development of quantum technologies is of full demanding challenges. From a technological point of view there is the difficulty of manipulating coherently quantum systems made of few particles while maintaining the quantum correlations. This implies that quantum systems have to be kept under very low temperatures and interaction with them has to be performed very fast in order to avoid decoherence \cite{Bacciagaluppi2016}.

A basic requirement for an effective quantum information processing system, quantum sensor or simulator is the ability to control the quantum state of the system at the individual level. The control of quantum spin systems can be addressed by means of the geometrical control theory. For instance, Khaneja et al.\ showed that finding sub-Riemannian geodesics on a quotient space of $SU(4)$ allows to obtain RF pulse trains for two-spin and three-spin NMR systems \cite{khaneja_sub-riemannian_2002} and also studied the associated numerical implementations \cite{khaneja_optimal_2005}. Also in \cite{moseley_geometric_2004} they study the control of quantum spin systems using a geometrical control approach, and the optimal control problem for blocks of quantum algorithms has been considered in \cite{schulte-herbruggen_optimal_2005} (see also \cite[Chps. 5,6]{dalessandro_introduction_2007}, \cite{bonnard_review_2012}, \cite{glaser_training_2015} and references therein for more recent reviews on geometric quantum control). Other recent approach to optimal control of coupled spin systems is described in \cite{delgado-tellez_optimal_2016}.

However, geometric control theory and its extension to optimal control problems suffers serious drawbacks when extended to genuine infinite dimensional quantum systems. Mainly because of the intrinsic  mathematical difficulties of infinite-dimensional geometry. Nevertheless, it has been applied to the finite dimensional approximations used to model the aforementioned quantum devices such as Ion Traps, NMR quantum computers and others. One of the sources of decoherence comes precisely by the neglection of the highest energy levels in order to perform the finite-dimensional approximations \cite{Garg99}.

There are not many results on controllability of infinite dimensional systems (see, e.g.\ \cite{beauchard_local_2005,beauchard_controllability_2006}, \cite{chambrion_controllability_2009} and references therein). As an alternative method to geometric control theory, there is the use of quadrature operators. There the usual approach is to associate quadrature operators to the problem and study the control of their dynamics (see for instance \cite{agarwal_quantum_2012}, \cite{carlini_time-optimal_2014}, \cite{genoni_optimal_2013}).

The quantum control at the boundary (QCB) method is a radically different approach to the problem of controlling the state of a qubit. Instead of seeking the control of the quantum state by directly interacting with it using external magnetic or electric fields, the control of the state will be achieved by manipulating the boundary conditions of the system. The spectrum of a quantum system, for instance an electron moving in a box, depends on the boundary conditions imposed on it. The typical situation is to consider either Dirichlet or Neumann boundary conditions. A modification of such boundary conditions modifies the state of the system allowing for its manipulation and, eventually, its control \cite{ibort_quantum_2010}. Addressing the problem from the genuine infinite dimensional setting provides a natural way of avoiding sources of decoherence.

The QCB paradigm has been used to show how to generate entangled states in composite systems by suitable modifications of the boundary conditions \cite{ibort_boundary_2014}. The relation of QCB and topology change has been explored in \cite{perez-pardo_boundary_2015} and recently used to describe the physical properties of systems with moving walls (\cite{facchi_moving_2016}, \cite{facchi_boundaries_2018}, \cite{facchi_quantum_2018}, \cite{facchi_self-adjoint_2018}, \cite{garnero_quantum_2018}), but in spite of its intrinsic interest some basic issues such as the QCB controllability of simple systems has never been addressed.

In developing the theory it will be shown first, by means of a suitable chosen time-dependent unitary transformation, that the variation of the boundary conditions of the system can be implemented as a time-dependent family of Hamiltonian operators, an idea that was already anticipated in \cite{perez-pardo_boundary_2015}. The particular instance of quasi-periodic boundary conditions will be worked out explicitly and it will be shown that the system reduces to a linear system similar to those studied by Chambrion et al.\ \cite{chambrion_controllability_2009}.

This article is organised as follows. In Section~\ref{sec:controllability} we review the notions of controllability in quantum systems that we will need to address the problem. The main difficulties will also be presented. In Section~\ref{sec:magneticlaplacian} we introduce the magnetic Laplacian. This provides a simple model where we will be able to implement the scheme of QCB and prove controllability rigorously. Section~\ref{sec:dynamics} and Section~\ref{sec:approximatecontrol} are devoted respectively to prove the well-posedness of the dynamics in the particular system of quantum control at the boundary considered and its approximate controlability.


\section{Control of Quantum Systems}\label{sec:controllability}

As stated in the introduction, one of the main objectives of the research presented in this article is to show that the paradigm of quantum control at the boundary is feasible. That is, we will prove controllability, in the sense that we are going to introduce later in this section, of a quantum system by means of modifications of the boundary conditions. Before doing that, let us review briefly some important concepts of the standard theory of control.

To fix the ideas in the context of Quantum Control, cf. \cite{dalessandro_introduction_2007}, consider the following setting. The space of pure states is given by the complex projective space , $\mathcal{P}(\H)$, of the separable Hilbert space  $\H$, \cite{GKM05, EMM10, CIMM15}. In what follows we will denote the norm and scalar product of the Hilbert space by the usual notation, i.e. $\norm{\cdot}$ and $\scalar{\cdot}{\cdot}$ respectively.

Evolution in a quantum system is governed in general by the time-dependent Schrödinger Equation. The final purpose of Control Theory is to study how to introduce an interaction into a system in order to be able to drive the state of the system from a given initial state to a desired target state. A simple, yet convenient, setting to define quantum control is to consider a time-dependent Hamiltonian of the form
\begin{equation}\label{eq:affinecontrol}
    H(t) = H_0 + \sum_{i= 1}^n f_i(t) H_i\,,
\end{equation}
where $H_0$ and $H_i$ are self-adjoint operators on the Hilbert space and where $f_i(t)\in \mathcal{C}$ are one variable functions on a convenient space of functions. The latter has to be specified and depends on the particular problem that one wants to address. Since Control Theory is devised ultimately to be applied to some concrete experimental setting, the limitations or restrictions to be imposed on the family of controls $\mathcal{C}$ will come from the experimental setup. For simplicity let us consider for the moment that $\mathcal{C} \equiv \mathcal{C}^\infty(\R)$, the space of smooth and real-valued functions. Given an initial state $\Psi_0\in\mathcal{P}(\H)$ and a target state $\Psi_T\in\mathcal{P}(\H)$, the problem of controllability consists on determining if there exists a choice of the functions $f_i(t) \in \mathcal{C}$ such that the solution of the time-dependent Schrödinger Equation is such that the initial state $\Psi_0$ is driven to the target state $\Psi_T$ in a time $T>0$. In order to give a more precise definition of controllability let us introduce the reachable set.

\begin{definition}
    Let $\Psi_0\in\mathcal{P}(\H)$, $f_i(t)\in \mathcal{C}$, $i= 1,\dots,n$ and let $\Psi(t)$ be the solution of the time dependent Schrödinger Equation
    $$i\frac{\d}{\d t}\Psi(t) = H(t) \Psi(t).$$
    The \textbf{reachable set} $\mathcal{R}_{\Psi_0}(T)$ of the state $\Psi_0$ at time $T\in\R$ is defined to be
    $$\mathcal{R}_{\Psi_0}(T) = \left\{  \Psi\in\mathcal{P}(\H)  \mid  \Psi = \Psi(t), t<T\in\R, \Psi(0) = \Psi_0, f_i(t)\in\mathcal{C}, i=1,\cdots,n \right\}.$$
\end{definition}
That is, the reachable set of the state $\Psi_0$ is the set of all those states that can be accessed starting at the state $\Psi_0$ under all the possible evolutions described by the family of controls. We postpone until later in this section the considerations of existence of solutions of this time-dependent Schrödinger Equation. For the definition of reachable set it is implicitly assumed that the initial value problem is well-posed. We are ready now to define the notion of exact controllability.
\begin{definition}
    Let $\mathcal{C}$ be a family of controls, let the quantum system defined by the space of states $\mathcal{P}(\H)$ and evolution determined by the time-dependent Schrödinger Equation
    $$i\frac{\d}{\d t}\Psi(t) = H(t) \Psi(t),$$
    with Hamiltonian $H(t) = H_0 + \sum_{i= 1}^n f_i(t) H_i$, $f_i(t) \in \mathcal{C}$, $i = 1,\cdots,n$. The quantum system is said to be \textbf{exactly controllable} if for all $\Psi_0\in \mathcal{P}(\H)$ one has that
    $$\bigcup_{T\in\R}\mathcal{R}_{\Psi_0}(T) = \mathcal{P}(\H).$$
\end{definition}
This notion is also called in the literature pure state controllability. We are only interested in the evolution of pure states $\Psi \in \mathcal{P}(\H)$, in contrast to the more general density states.

Let us say now that the problem of controllability is a problem of existence of controls such that any target state can be achieved. The problem of (optimal) determination of the controls will not be considered here.

In general, the quantum systems are defined on infinite-dimensional Hilbert spaces. Moreover, typically the Hamiltonians are unbounded operators acting on the Hilbert space of the system. Unbounded operators are not continuous operators on the Hilbert space and therefore, existence of solutions of the time-dependent Schrödinger Equation is compromised. For instance, the domains of the operators may depend also on  time, and the range of the operators may not preserve the domains. These facts introduce a set of stringent conditions on the families of available Hamiltonians, and thus of available controls, in order to define well-posed control problems. One of the aims of this article is to show that the setting of quantum control at the boundary is feasible. In particular this implies guaranteeing the existence of solutions of the evolution equation.

Quantum systems of mechanical type are governed by Hamiltonians defined by differential operators on Riemannian manifolds. Typically, the Laplace-Beltrami operator or other second order differential operator related to it. If the Riemannian manifold has boundaries those operators are in general symmetric operators but not self-adjoint, cf. \cite{reed_methods_1975} or \cite{ibort_self-adjoint_2015b} and references therein for an introduction to the topic. Each self-adjoint extension describes a different physical situation. Consider, for example, the case of the Laplace operator on a compact interval. One can consider Dirichlet boundary conditions or Neumann boundary conditions. These two operators define two completely different self-adjoint extensions of the same operator and thus describe completely different evolutions. The space of self-adjoint extensions of a symmetric, second order differential operator (in any dimension) can be characterised by certain families of boundary conditions, cf. \cite{Grubb1968, ibort_self-adjoint_2015} and references therein.

We will consider the use of these spaces of boundary conditions as spaces of controls. This idea was firstly introduced in \cite{ibort_quantum_2010}. The appearance of the controls in the Hamiltonian are now more subtle than they are in Eq.~\eqref{eq:affinecontrol} since they will not appear directly in the functional form of the operator, but will appear in the boundary conditions that define the different domains of the operators at every instant of time. That is, we are going to consider families of Hamiltonians $\left(H, \DD(f_i)\right)$, $f_i(t) \in \mathcal{C}$, where the space of controls $\mathcal{C}$ is now the space of self-adjoint extensions (or a subset of it) of the symmetric operator $H$.

From these previous considerations it follows that the setting of quantum control at the boundary requires of infinite-dimensional Hilbert spaces and unbounded operators. Unfortunately, the usual notions of control introduced at the beginning of this section are not suitable to handle the infinite dimensional situation. In particular, they turn out be too strict and there is the need to introduce a notion of controllability that is slightly weaker. Consider the quantum control system defined by the Harmonic oscillator over the real line
    $$H_0 = -\frac{1}{2}\frac{\d^2}{\d x^2} + \frac{1}{2} x^2\,,\quad H_1 = x\,, \quad \mathcal{C}\equiv \mathcal{C}^\infty(t)\,,$$
such that $H(t) = H_0 + f(t)H_1$, $f(t) \in \mathcal{C}$. This quantum control system is not exactly controllable, see for instance \cite{mirrahimi_controllability_2004}. However, every finite dimensional truncation up to the first $n$ lowest eigensates, whose Hamiltonians are now given by Hermitean matrices \hbox{$\tilde{H}_0,\, \tilde{H}_1 \in M(\C)^{n\times n}$} is exactly controllable, \cite{ramakrishna_controllability_1995}. This situation motivates the definition of approximate controllability.
\begin{definition}
    Let $\Psi_0, \Psi_T \in \mathcal{P}(H)$. Let $B_\epsilon(\Psi_T)$ be the ball of radius $\epsilon>0$ centred at $\Psi_T$. We will say that a quantum system is \textbf{approximately controllable} if for every $\epsilon>0$ there exists a $T>0$ such that
    $$\mathcal{R}_{\Psi_0}(T)\cap B_\epsilon(\Psi_T) \neq \emptyset.$$
\end{definition}
That is, a quantum system is approximately controllable if there is a finite time $T$ such that the reachable set $\mathcal{R}_{\Psi_0}(T)$ of the state $\Psi_0$ intersects with a neighbourhood of radius $\epsilon$ of the state $\Psi_T$. Therefore, one can come as close to the target state $\Psi_T$ as desired. It is remarkable that approximate controllability has been proven , cf. \cite{chambrion_controllability_2009}, for linear systems with one control ($n= 1$) under suitable assumptions on the spectral properties of the operators $H_0$ and $H_1$. On Section~\ref{sec:approximatecontrol} we will rely on that result to prove controllability for a particular instance of quantum control at the boundary. We should mention here that more general notions of controllability, suitable for quantum systems, are also possible, e.g. \cite{IPP09}.


\section{Magnetic Laplacian} \label{sec:magneticlaplacian}
    During the rest of this work we will concentrate in one particular class of quantum systems, namely magnetic Laplacians in one dimension. The reason behind this choice is twofold. On one hand these systems are simple enough such that we will be able to prove rigorously the existence of dynamics and to address the boundary controllability problem. On the other hand, this simple system can be implemented physically, thus opening an interesting path to devise applications of the scheme of quantum control at the boundary to quantum computation and quantum information. Let $L\subset\R$ be a compact interval that for conveniece we will consider to be $L = [0,l]$. The Hamiltonian of the magnetic Laplacian takes the form:
        \[ H = -\left( \frac{d}{dx} - i A(x) \right)^2 \eqqcolon -D^2, \]
        where $A\in\mathcal{H}^1(L)$ is a function in the Sobolev space of order 1 and is called the magnetic potential.

    From its definition it can be seen the similarity of this Hamiltonian with the Laplace operator. This operator describes the so called \emph{minimal coupling} of an electrically charged particle with a magnetic potential. This justifies the name of magnetic Laplacian. It is a second order differential operator and we need to determine a domain for it in order to have it well defined. Following \cite{Kochubei1975,asorey2005global,ibort_self-adjoint_2015} we will identify the domains of self-adjointness by looking for maximal domains where the boundary term of Green's formula vanishes identically. This boundary term reads in this case:

            \[
            i{\Sigma}(\Phi,\Psi) \coloneqq i( \langle \Psi, D^2 \Phi \rangle - \langle D^2 \Psi, \Phi \rangle)
            = \langle \veval{\Psi} + i\oveval{D \Psi}, \veval{\Phi} + i\oveval{D \Phi} \rangle_{\partial L} - \langle \veval{\Psi} - i\oveval{D \Psi}, \veval{\Phi} - i\oveval{D \Phi} \rangle_{\partial L}.
            \]
    The underline notation stands for restrictions to the boundary, while the arrows over the symbols mean that the restriction to the boundary is taken having into account the orientation. That is, derivatives are taken with orientation pointing outwards to the boundary as well as the restriction to the boundary of the potentials (they are one-forms evaluated on the normal vector to the boundary). The subindex $\partial L$ means that it is considered the scalar product of the Hilbert space induced at the boundary of $L$. Therefore, cf. \cite{asorey2005global}, the self-adjoint extensions of $D^2$ are parametrized by an unitary operator $U \in \U(\mathcal{L}^2(\partial L))$ with
            \[
                \dom D^2_U= \{\Phi\in \mathcal{H}^2(L) :
                \veval{\Phi} - i \oveval{D\Phi} = U (\veval{\Phi} + i \oveval{D\Phi})\},
            \]
            where $\mathcal{H}^2(L)$ is the Sobolev space of order 2.

            It is going to be convenient for the next section to keep in mind the following well-known property about magnetic Laplacians (see, e.g., \cite{kostrykin2003quantum} for a more detailed study of this properties).

    \begin{proposition}\label{prop:magnetic-Laplacian-equiv}
        Let $D^2_U$ be a self-adjoint extension of the magnetic Laplacian associated to a vector potential $A$. Then, for any $\tilde{A}$ there exists a self-adjoint extension of the associated magnetic Laplacian, $\tilde{D}^2_V$, and an isometry $T$ on $\mathcal{L}^2(L)$ mapping $\dom \tilde{D}^2_V$ into $\dom D^2_U$ such that
        \[ T^{-1} D^2_U T = \tilde{D}^2_V. \]
        Moreover, $V = \veval{T}^{-1}U\veval{T}$ with $\veval{T}$ the restriction to the boundary of $T$,i.e. the operator $\veval{T} : \mathcal{L}^2(\partial L) \mapsto \mathcal{L}^2(\partial L)$ such that $\veval{T}\,\veval{\Phi} = \veval{T\Phi}$ for any $\Phi\in\mathcal{H}^2(L)$.
    \end{proposition}
    \begin{proof}
        As we already said, the magnetic vector potential is taken to be continuous and therefore, by the Poincaré Lemma, there exists $\chi: L \to \R$ differentiable such that $\chi' = A - \tilde{A}$. Let $T$ denote the multiplication map defined by
        \begin{equation}\label{eq:def-T}
            T: \Phi \in \mathcal{L}^2(L) \mapsto e^{i \chi } \Phi  \in \mathcal{L}^2(L).
        \end{equation}
        It follows directly from this definition that $T$ is an isometry on $\mathcal{L}^2(L)$. Using the product rule, it is easy to check that
        \[
            \left(\frac{d}{dx} - iA(x)\right) T \Psi = T\left(\frac{d}{dx} - i\tilde{A}(x)\right)\Psi.
        \]
        Evaluating at the boundary, it follows
        \begin{equation*}\label{eq:prop1.1-1}
            \veval{D_U T \Psi} = \veval{T \tilde{D}_V \Psi} = \veval{T\vphantom{\Psi}}\,\veval{\tilde{D}_V \Psi},
        \end{equation*}
        where $\veval{T}$ is the diagonal matrix $\veval{T} = \diag(\{e^{i\underline{\chi} (v)}\}_{ v\in \partial L})$.

        Using this, it is straightforward to show that for any $\Phi \in \dom D^2_U$, $\Psi = T^{-1}\Phi$ is in $\dom \tilde{D}^2_V$, if $V = \veval{T}^{-1}U\veval{T}$. Moreover,
        \[ D^2_U T \Psi = T \tilde{D}^2_{\veval{T}^{-1}U\veval{T}}, \]
        which concludes the proof.
    \end{proof}

    As a consequence of this property, we show the following result which will allow us to consider constant vector potentials.

    \begin{corollary}\label{corol:magnetic-Laplacian-constant}
         Every self-adjoint extension of a magnetic Laplacian $D^2_U$, associated with a potential $A$, is equivalent to one associated with a constant potential $\tilde{A}$ such that $\dom D^2_U = \dom \tilde{D}^2_U$.
    \end{corollary}
    \begin{proof}
         Let $l$ denote the length of the interval, i.e.\ $\int_L\d x = l$.
         Take $\tilde{A}  = l^{-1} \int_L A (x) \, \d x$ and $\chi (x) = \int_0^{x} (A (x) - \tilde{A} ) \, \d x$. Define
             \[ \tilde{D}^2_U = \left(\frac{d}{dx} - i \tilde{A}\right)^2 \]
         and $T$ as in Proposition~\ref{prop:magnetic-Laplacian-equiv}; it follows that $\veval{T} = \mathbb{I}_{2\times 2}$ and therefore by Proposition~\ref{prop:magnetic-Laplacian-equiv}
             \[ T^{-1} D^2_U T = \tilde{D}^2_U. \]
    \end{proof}

            Finally, it follows straightforwardly from the previous corollary the next result, wich will be the base result for our main purpose to prove controllability at the boundary.

            \begin{corollary}\label{corol:magnetic-standard-equiv}
Let $\Delta$ stand for the Laplacian, i.e. $D^2$ with $A\equiv0$. For every magnetic Laplacian, $D_U^2$, there is an equivalent self-adjoint extension of the Laplacian. Moreover, if $T$ is the multiplication operator defined on Equation \eqref{eq:def-T} with $\chi$ such that $\chi' = A$, then
                \[ T^{-1} D^2_U T = \Delta_{\underline{T}^{-1}U\underline{T}}. \]
            \end{corollary}

    Among the possible unitary operators $U\in\mathcal{U}(\partial L) \simeq \mathcal{U}(\mathbb{C}^2)$ that one can consider, there are different relevant particular choices. It is important to mention that $U=\mathbb{I}_{2\times 2}$ defines Neumann boundary conditions while $U = -\mathbb{I}_{2\times 2}$ defines Dirichlet boundary conditions. A simple calculation shows that
        $ U = \begin{bmatrix}
            0 & 1 \\ 1 & 0
            \end{bmatrix}
        $
    defines periodic boundary conditions, i.e.\ \hbox{$\Phi(0) = \Phi(l)$}, \hbox{$(D\Phi)(0)= (D\Phi)(l)$}.
    The previous corollaries allow us to define the family of boundary conditions that we will use for the implementation of quantum control at the boundary.

    \begin{definition}
        Let $D_U$ be a magnetic Laplacian on the interval $L$ with periodic boundary conditions and with a constant magnetic potential. Let $T$ be the multiplication operator defined by Equation~\eqref{eq:def-T} with $\chi:L\mapsto\R$ such that $\frac{\d\chi}{\d x}=A$. Then $V= \veval{T}^{-1}U\veval{T}$ defines quasi-periodic boundary conditions.
    \end{definition}

A simple computation shows that the unitary operators appearing in this definition are:
                \[
                    V = \begin{bmatrix}
                        0            &    e^{-iAl}    \\
                        e^{iAl}    &    0
                    \end{bmatrix},
                    \qquad
                    \veval{T} = \begin{bmatrix}
                        1 & 0 \\
                        0 & e^{iA}
                    \end{bmatrix}e^{i b} ,
                \]
            where $b$ is the constant of integration in the definition of the function $\chi(x)$.

The way in which we are going to make use of the result in Corollary~\ref{corol:magnetic-standard-equiv} is as follows. As the quantum control system we will take a free particle moving in the interval $L$. That is, the family of Hamiltonians is taken to be the standard Laplacian or, equivalently, the magnetic Laplacian with $A\equiv 0$. As explained earlier in this section, these operators are not well defined until we fix the corresponding domains. Each operator in this family is going to be characterised by a different quasi-periodic boundary condition. By Corollaries~\ref{corol:magnetic-Laplacian-constant} and \ref{corol:magnetic-standard-equiv} each of these systems is unitarily equivalent to a magnetic Laplacian with constant magnetic potential $A$ and periodic boundary conditions. We stress here that by constant we mean that the potential has the same value, independent of the point of the interval. From now on we will use the same symbol $A$ to denote the constant magnetic potential $A\in\H^1(L)$ and its value $A\in\R$. We are going to consider that the the constant $b=0$. The transformation $T$ of Equation~\eqref{eq:def-T} is defined in this case by the function
$$\chi(x) = A x.$$

We want now to implement the scheme of quantum control at the boundary. This means that we are going to use the parameter $A$ defining the boundary condition as our control and we will suppose that now $A=A(t)$ is a function of time. At every instant of time it will still be a constant magnetic potential along the interval $L$, but its magnitude will depend on time and constitute our control parameter.

Thus, we consider a quantum control system whose Hamiltonians are standard Laplacians with time-dependent quasi-periodic boundary conditions such that
            \begin{equation}\label{eq:quasiperiodic-timedependent}
                \veval{\Psi} (0) = e^{-i\veval{\chi} (l,t)} \veval{\Psi}(l),
            \end{equation}
    where now
        \begin{equation}\label{eq:timepotential}
            \chi(x,t) = A(t) x
        \end{equation}
    forms a family of functions from $L$ to $\R$. One should notice that the time dependence of these Hamiltonians is subtle: usually one faces the problem where $\dom H(t)$ does not depend on time but the explicit, functional form of $H(t)$ does, while here we have $-\Delta$ for every $t$ and $\dom \Delta$ varying with time. That is, we are considering at each time a different self-adjoint extension of the Laplacian on our interval $L$.

    Therefore, as anticipated in the previous section, looking for solutions of the time-dependent Schrödinger equation is harder than in the most common situations. However, based on the equivalence established in this section we will be able to transform these problems into equivalent ones with Hamiltonian $H(t)$ such that $\dom H(t)$ remains independent of $t$ and time dependence appears explicitly in the form of $H(t)$.

In summary, we are interested in the following control problem.

    \begin{definition}\label{def:quase-periodic-BCS}
        Consider the compact interval $L=[0,l]$. The \textbf{boundary control system associated to $L$} is the family of quantum Hamiltonians defined by the Laplace operator and domains given by quasi-periodic boundary conditions $\dom \Delta_{U(t)}$, with
        $$U(t) = \begin{bmatrix}
                           0 & e^{-iA(t)l} \\ e^{iA(t)l} & 0
                       \end{bmatrix}.$$
    \end{definition}


\section{Existence of Dynamics in Boundary Control Systems}\label{sec:dynamics}

The aim of this section is to study the dynamics of a boundary control system as defined in Definition~\ref{def:quase-periodic-BCS}. It will turn out that the dynamics will be well defined if the control function $A:\R \mapsto \mathcal{H}^1(L)$ varies smoothly with time.
Quantum systems' evolution is given by a Hamiltonian operator $H(t)$, which in the most general setting depends itself on the time $t$, and according to Schrödinger equation
        \begin{equation}\label{eq:schrodinger}
            i \Dt \Psi(t) = H(t) \Psi(t).
        \end{equation}
        In the case we are interested on $H(t)$ is a family of differential operators on $\mathcal{L}^2(L)$ and $\Psi(t)$ is a curve in the state space $\mathcal{P}(\mathcal{H})$.

        Concerning the existence of solutions for the Schrödinger equation with a given Hamiltonian, there are several results establishing conditions for solutions to exist \cite{kisynski_sur_1964, reed_methods_1975}. It is customary to search for solutions using the idea of unitary propagators, which are families of operators which allow us to write the solution of the Schrödinger equation with initial state $\Psi_s$ at $t=s$ as $\Psi(t) = U(t,s)\Psi_s$ for $t > s$. A proper definition of a unitary propagator would be as follows:

        \begin{definition}\label{def:unitary-propagator}
            A two-parameter family of unitary operators $U(s,t)$, with $s,t \in \R$, that satisfies:
            \begin{enumerate}[label=\textit{(\roman*)},nosep]
                \item $U(r, s)U(s, t) = U(r, t)$
                \item $U(t,t) = I$
                \item $U(s, t)$ is jointly strongly continuous in $s$ and $t$
            \end{enumerate}
            is called a \textbf{unitary propagator}.
        \end{definition}

        After unitary propagators are introduced, the existence of the solutions for the associated Cauchy problems is equivalent to the existence of a unitary propagator for the Eq.~\eqref{eq:schrodinger}. For the most general setting, in which $\dom H(t)$ varies with $t$, J. Kisyński gave conditions that $H(t)$ must satisfy for the unitary propagator to exist \cite{kisynski_sur_1964}. However, we will be interested in the less general case in which $\DD = \dom H(t)$ is the same for every $t$ and thus it is enough to consider a less general result by M. Reed and B. Simon \cite[\S X.12]{reed_methods_1975}. Instead of treating the case of families of self-adjoint operators, they study the more general case of families of generators of contraction semigroups, which can be directly applied to the case of families of self-adjoint operators since for $H$ self-adjoint, $\pm iH$ is the generator of a contraction semigroup (see Theorem X.47a and Example 1 on \S{X}.8 of \cite{reed_methods_1975}).

    Let $S(t)$ denote a family of generators of a contraction semigroup. For such a case, Reed and Simon define an approximation for the propagator $U(t,s)$ solving the equation
        \[
            \Dt \varphi(t) = -S(t) \varphi(t), \qquad \varphi(s) = \varphi_s
        \]
    in the following way. First there is considered a partition of the time interval, taking a generator which is constant on each element of the partition and providing conditions ensuring that it converges to the solution. If, for example, the time interval we are interested in is $I = [0, 1]$, they take the partition made of $k$ elements $I_j = [\frac{j-1}{k}, \frac{j}{k}]$, $1 \leq j \leq k$, and define the approximate propagator
        \begin{equation}\label{eq:RSpropagator}
            U_k(t, s) = \begin{cases}
                \exp\left(-i(t-s) S\left( \frac{j - 1}{k} \right)\right)  & \text{if}\quad \frac{j - 1}{k} \leq s \leq t \leq \frac{j}{k} \\
                U_k(t, \frac{j-1}{k})
                                  U_k(\frac{j-1}{k}, \frac{j-2}{k})
                                  \cdots
                                  U_k(\frac{j-l}{k}, s) & \text{if}\quad \frac{j-(l+1)}{k} \leq s \leq \frac{j-l}{k} \leq \frac{j-1}{k} \leq  t \leq \frac{j}{k}.
                              \end{cases}
         \end{equation}

    That is, if $s,t$ lie in the same interval $I_j$ they consider the evolution operator given by the action of the contraction semigroup generated by $S(\frac{j-1}{k})$ and if $t,s$ lie in different intervals, they use the product property of the unitary propagator to define it.

Before stating the Theorem by M. Reed and B. Simon let us prove the following result that allows to treat the boundary control problem as a time dependent problem with fixed domain. Following the ideas exposed in the previous section, we can find a natural equivalence between a boundary control system and a magnetic controlled one:

\begin{proposition}\label{prop:equivmagneticLaplacian}
    Every boundary control system is (unitarily) equivalent to a magnetic control system, that is, a system whose evolution is given by the Hamiltonian
    \[ H(t) = -\left[\left(\Dx - i A(t)\right)^2 +  A'(t) x \right] \]
    with periodic boundary conditions and controls $A : I \subset \R \mapsto \mathcal{H}^1(L)$, where $I$ is some compact interval and $\mathcal{H}^1(L)$ is the Sobolev space of order 1 on the interval $L$.
\end{proposition}
\begin{proof}
    Take the family of unitary transformations $T(t)$ as in \eqref{eq:def-T} with $\chi = \chi(t)$:
    \[
        T(t): \Psi \in \mathcal{L}^2(L) \mapsto e^{i \chi (t)} \Psi \in \mathcal{L}^2(L).
    \]
    Define $\Phi(t) = T(t)\Psi(t)$. The chain rule implies
    \[ \Dt \Phi(t) = \D{T}{t}(t) \Psi(t) + T(t) \Dt \Psi(t), \]
    where the derivatives of the operators have to be understood in the strong operator topology sense. Using the Schrödinger equation for $\Psi$, cf. Eq.~\eqref{eq:schrodinger}, and the definition of $T(t)$, we have
    \[ i \Dt \Phi(t) = \left[- \Dt \chi(t) - T(t) \Delta T(t)^{-1} \right] \Phi(t). \]
    Take $\chi(t)=A(t)x$ as in Equation~\eqref{eq:timepotential} and remember that we are assuming that the integration constant is $b=0$.
    Thus, we have $\chi(t) = A(t) x$ and $T(t)\Delta T(t)^{-1} = -\left(\Dx - iA(t)\right)^2 \eqqcolon -D^2$. Thus we have finally that
    \[
        i \Dt \Phi(t) = -\left[\left(\Dx - i A(t)\right)^2 + A'(t) x \right] \Phi(t),
    \]
    with periodic boundary conditions for every $t\in I$:
    \[ \left\{ \begin{alignedat}{2}
        &\Phi(0) = \Phi(l)\\
        &(D\Phi)(0) = (D\Phi)(l) \Leftrightarrow \frac{\d\Phi}{\d x}\biggr|_{x=0} = \frac{\d\Phi}{\d x}\biggr|_{x=l}.
    \end{alignedat} \right. \]
\end{proof}

            Notice that the equivalence in the last condition follows because we are considering magnetic potentials that are constant on the interval $L$. This proposition shows how to treat the boundary control system applying a unitary transformation which leads to a magnetic controlled system, where the time-dependence of the Hamiltonian's domain has been removed. \\

        For each $t\in \R$ let  $S(t): \DD \subset \H \mapsto \H$ be the generator of a contraction semigroup, densely defined on $\DD$. Notice that we are assuming that the domain $\DD$ remains fixed for every $t$. Let $\rho(S(t))$ denote the resolvent set of the operator $S(t)$ and assume that $0 \in \rho(S(t))$ for all $t\in\R$. For convenience of the notation it is defined a two-parameter family of operators
        \[ C(t, s) = S(t)S(s)^{-1} - I. \]
        Note that $0 \in \rho(S(t))$ for all $t$ implies that $S(t)$ is a bijection of $\DD$ onto $\H$, and therefore $C(t, s)$ it is a bounded operator by the Closed Graph Theorem.  Moreover, for every $\Phi \in \H$ and every $s\in\R$ there exists $\Psi \in \DD$ such that $\Phi = S(s) \Psi$. Thus, for that $\Phi$,
        $$C(t,s) \Phi = S(t) \Psi - S(s)\Psi.$$
        That is, studying the behaviour of $C(t,s)\Phi$ for any $\Phi \in \H$ can be understood as studying that of $S(t) \Psi - S(s) \Psi$ for any $\Psi \in \DD$.
        In order to prove existence of dynamics of the boundary control system we are going to use the next result by M. Reed and B. Simon in what follows.

        \begin{theorem}[M. Reed and B. Simon, {\cite[Thm. X.70]{reed_methods_1975}}]\label{thm:reed-simon}
            Let $\H$ be a Hilbert space and let $I$ be an open interval in $\R$. For each $t \in I$, let $S(t)$ be the generator of a contraction semigroup on $\H$ so that $0 \in \rho(S(t))$ and
            \begin{enumerate}[label={\it (\alph*)},nosep]
                \item\label{enum:reedSimon-i} The $S(t)$ have common domain $\DD$.
                \item\label{enum:reedSimon-ii} For each $\Phi \in \H, (t - s)^{-1} C(t, s)\Phi$ is uniformly strongly continuous and uniformly bounded in $s$ and $t$ for $t \neq s$ lying in any fixed compact subinterval of $I$.
                \item\label{enum:reedSimon-iii} For each $\Phi \in \H$, $C(t)\Phi = \lim_{s \nearrow t} (t - s)^{-1}C(t, s)\Phi$ exists uniformly for $t$ in each compact subinterval of $I$ and $C(t)$ is bounded and strongly continuous in $t$.
            \end{enumerate}
        Then for all $s \leq t$ in any compact subinterval of $I$ and any $\Phi \in \H$,
        \[ U(t, s) \Phi = \lim_{k \to \infty} U_k(t, s)\Phi \]
        exists uniformly in $s$ and $t$, where $U_k(t, s)$ is given by Equation~\eqref{eq:RSpropagator}. Further, if $\Phi_s \in \DD$, then $\Phi(t) = U(t, s)\Phi_s$ is in $\DD$ for all $t$ and satisfies
        \[ \Dt \Phi(t) = -S(t) \Phi(t), \qquad \Phi(s) = \Phi_s \]
        and $\norm{\Phi(t)} \leq \norm{\Phi_s}$ for all $t \geq s$.\\
        \end{theorem}

        The rest of this section is devoted to prove that the family of magnetic Laplacians of Proposition~\ref{prop:equivmagneticLaplacian} meets the conditions of the Theorem~\ref{thm:reed-simon}.

        We are going to work with Hamiltonians whose time-dependent structure can be written as
        \[ H(t) = \sum_{i = 1}^n f_i(t) H_i, \]
        with $f_i: \R \to \R$ containing all the time dependence and $H_i$ being constant symmetric operators. Applying Theorem~\ref{thm:reed-simon} to this type of Hamiltonians is the purpose of Theorem~\ref{thm:timedependent-linearCombination-Hamiltonian}, which establishes sufficient conditions to be fulfilled so that the existence of a unitary propagator is guaranteed.

        \begin{theorem}\label{thm:timedependent-linearCombination-Hamiltonian}
            Let $\{H_i\}_{i=1}^n$ be a family of symmetric operators densely defined on $\mathcal{D}\subset{\H}$ and let $f_i: I \subset \R \to \R$ be real valued functions for $1 \leq i \leq n$. Define the time-dependent operator
            \[
                H(t) = \sum_{i = 1}^n f_i(t) H_i, \qquad
                \dom H(t) = \DD.
            \]
            If it holds
            \begin{enumerate}[label=\textit{(\roman*)},nosep]
                \item\label{enum:H1-thm-linCombHamiltonian} $H(t)$ is self-adjoint for all $t\in I$,
                \item\label{enum:H2-thm-linCombHamiltonian} $f_i \in C^1(I)$ for every $i$, and
                \item\label{enum:H3-thm-linCombHamiltonian} for every $i$ there exists a $K > 0$ (not depending on $t$) such that for every $\Psi \in \mathcal{D}$, {$$\norm{H_i \Psi} \leq K (\norm{H(t)\Psi} + \norm{\Psi})$$} for every $t \in I$.
            \end{enumerate}
            Then, there exists a strongly differentiable unitary propagator $U(t,s)$ with $s,t \in I$ such that, for any $\Psi_s \in \DD$, $\Psi(t) = U(t,s)\Psi_s$ satisfies
            \[ \Dt \Psi(t) = -i H(t) \Psi(t), \qquad \Psi(s) = \Psi_s. \]
        \end{theorem}

        Before we introduce the proof it is useful to introduce the following lemmas.

        \begin{lemma}\label{lemma:1}
            Let $H(t)$ be as in Theorem~\ref{thm:timedependent-linearCombination-Hamiltonian} and define $S_i = iH_i$ and $\tilde{S}(t) = iH(t) + I$. Then, for every $\Phi \in \mathcal{H}$, there exists $K>0$, independent of $\Phi$ and $t$, such that
            \[
                \norm{S_i \tilde{S}(t)^{-1} \Phi} \leq K \norm{\Phi}.
            \]
        \end{lemma}
        \begin{proof}
            Since $H(t)$ is self-adjoint, the spectrum of $\tilde{S}(t)$ is a subset of $i \mathbb{R} + 1 = \{i \alpha + 1: \alpha \in \mathbb{R}\} \subset \mathbb{C}$. Thus $\tilde{S}^{-1}$ is bounded and maps $\mathcal{H}$ onto $\mathcal{D}$.

            That said, this lemma is a direct consequence of hypothesis \ref{enum:H3-thm-linCombHamiltonian} of Theorem~\ref{thm:timedependent-linearCombination-Hamiltonian}. For every $\Psi \in \mathcal{D}$,
            \[
                \norm{S_i \Psi} = \norm{H_i \Psi} \leq K (\norm{H(t) \Psi} + \norm{\Psi})
                = K\left(\norm{\tilde{S}(t)\Psi - \Psi} + \norm{\Psi} \right) \leq K \left(\norm{\tilde{S}(t)\Psi} + 2\norm{\Psi} \right)
            \]
            for every $t$. Since $\tilde{S}(t)^{-1}\Phi \in \mathcal{D}$, for every $\Phi \in \mathcal{H}$ it holds:
            \[
                \norm{S_i \tilde{S}(t)^{-1}\Phi}
                \leq K \left(\norm{\Phi} + 2\norm{\tilde{S}(t)^{-1}\Phi} \right).
            \]
            The distance from 0 to $\sigma(\tilde{S}(t))$ is at least 1, and thus $\norm{\tilde{S}(t)^{-1}} \leq 1$. Hence, renaming the constant we get
            \[
                \norm{S_i \tilde{S}(t)^{-1}\Phi} \leq K \norm{\Phi}.
            \]
        \end{proof}
        \begin{lemma}\label{lemma:2}
            Let $H(t)$ be as in Theorem~\ref{thm:timedependent-linearCombination-Hamiltonian} and define $S_i = iH_i$ and $\tilde{S}(t) = iH(t) + I$. Then, for $s,t$ lying in a compact subset of $I$ and for every $\Phi \in \mathcal{H}$, $\lim_{t \to s} \tilde{S}(t)\tilde{S}(s)^{-1} \Phi = \Phi$ uniformly on $s$.
        \end{lemma}
        \begin{proof}
            We need to prove the limit uniformly on $s$, i.e., that
            \[
                \lim_{t \to s} \norm{\tilde{S}(t)\tilde{S}(s)^{-1} \Phi - \Phi} = 0
            \]
            uniformly on $s$. We have
            \[
                \norm{\tilde{S}(t)\tilde{S}(s)^{-1} \Phi - \Phi}
                = \norm{\sum_{i=1}^n[f_i(t) - f_i(s)]S_i \tilde{S}(s)^{-1} \Phi}
                \leq \sum_{i=1}^n |f_i(t) - f_i(s)| \norm{S_i \tilde{S}(s)^{-1}\Phi}.
            \]
            Now, using Lemma~\ref{lemma:1} one gets
            \[
                \norm{\tilde{S}(t)\tilde{S}(s)^{-1} \Phi - \Phi}
                \leq K \sum_{i=1}^n |f_i(t) - f_i(s)| \norm{\Phi}.
            \]
            Hence, $\lim_{t \to s} \norm{\tilde{S}(t)\tilde{S}(s)^{-1} \Phi - \Phi} = 0$ uniformly on $s$ since every $f_i$ is uniformly continuous on $I$ (which follows from \ref{enum:H2-thm-linCombHamiltonian} and the fact that we are considering a fixed, compact subset of $I$).
        \end{proof}

        \begin{proof}[Proof of Theorem~\ref{thm:timedependent-linearCombination-Hamiltonian}]
            This theorem is a consequence of Theorem~\ref{thm:reed-simon} and the fact that \hbox{$S(t) = i H(t)$} is the generator of a contraction semigroup by Hille--Yoshida theorem (see \cite{reed_methods_1975}, Theorem X.47a and Example 1 on \S{X}.8).
            In order to apply Theorem~\ref{thm:reed-simon} we need to have $0 \in \rho(iH(t))$ for every $t$, which is not satisfied in general. However, since $H(t)$ is self-adjoint $i \in \rho(H(t))$ for every $t$ and therefore $-1 \in \rho(S(t))$ which implies $\tilde{S} = S(t) + I$ has 0 in its resolvent set. Note that if $\Phi(t) = \tilde{U}(t,s) \xi$ satisfies
            \[ \Dt \Phi(t) = -\tilde{S}(t) \Phi(t), \qquad \Phi(s) = \xi,\]
            then $\Psi(t) = U(t,s) \xi$ with $U(t,s) \coloneqq \tilde{U}(t,s) e^{-i(s-t)}$ satisfies, by the product rule,
            \[ \Dt \Psi(t) = - S(t) \Psi(t), \qquad \Psi(s) = \xi. \]
            Thus, existence of $\tilde{U}(t,s)$ with the properties in the statement of Theorem~\ref{thm:timedependent-linearCombination-Hamiltonian} guarantee the existence of $U(t,s)$ with the same properties.

            Hence, it is enough to show that $\tilde{S}(t)$ satisfies the hypothesis of Theorem~\ref{thm:reed-simon}. It is clear that $\tilde{S}(t)$ can be written as
            \[
                \tilde{S}(t) = I + i\sum_{j=1}^{n} f_j(t) H_j
                = \sum_{j=1}^{n+1} f_j(t) S_j
            \]
            with $S_i = iH_j$, $f_{n+1} = 1$ and $S_{n+1} = I$. Also is easy to check using Hille--Yoshida theorem that $\tilde{S}(t)$ is the generator of a contraction semigroup.

            Hypothesis \ref{enum:reedSimon-i} of Theorem~\ref{thm:reed-simon} is satisfied since, by definition, every $H(t)$ (and therefore $\tilde{S}(t)$) has the same domain $\DD$.

            Regarding \ref{enum:reedSimon-ii} and \ref{enum:reedSimon-iii}, it is useful to write
            \begin{equation}\label{eq:thm-5.4-1}
                (t-s)^{-1}C(t,s) = (t-s)^{-1} [\tilde{S}(t) - \tilde{S}(s)] \tilde{S}(s)^{-1}
                = \sum_{i=1}^{n} \frac{f_i(t) - f_i(s)}{t - s} S_i \tilde{S}(s)^{-1}.
            \end{equation}
            For convenience, let us denote $g_i(t,s) = \frac{f_i(t) - f_i(s)}{t-s}$, which clearly is $C^1$ in $s$ and $t$ for $t \neq s$ in $I$. Moreover, for $s\neq t$ lying in any fixed compact subinterval of $I$, $g_i$ is uniformly continuous because $f_i(t)$ is $C^1(I)$.

            From the previous equation it follows that for $\Phi\in\H$
            \[
                \norm{(t-s)^{-1}C(t,s)\Phi}
                \leq \sum_{i=1}^{n+1} |g_i(t,s)| \norm{S_i \tilde{S}(s)^{-1}\Phi}
                \leq K \sum_{i=1}^{n+1} |g_i(t,s)| \norm{\Phi},
            \]
            where we have used Lemma~\ref{lemma:1} in the last inequality. For $s\neq t$ lying in any fixed compact subinterval of $I$, $|g_i(t,s)|$ is bounded uniformly on $s$ and $t$ since it is continuous and thus $\norm{(t-s)^{-1}C(t,s)\Phi}$ is uniformly bounded for such $s,t$.

            For the uniform strong continuity with respect to $t$, it is clear that
            \[\begin{alignedat}{2}
                \norm{(t_0-s)^{-1}C(t_0,s)\Phi - (t-s)^{-1}C(t,s)\Phi}
                & \leq \sum_{i=1}^{n+1} |g_i(t_0,s) - g_i(t,s)| \norm{S_i \tilde{S}(s)^{-1}\Phi} \\
                & \leq K \sum_{i=1}^{n+1} |g_i(t_0,s) - g_i(t,s)| \norm{\Phi}
            \end{alignedat}\]
            and, thus, uniform continuity of $t \mapsto g_i(t,s)$ implies uniform strong continuity of the operator-valued function $t \mapsto (t-s)^{-1}C(t,s)$.

            On the other hand, regarding uniform strong continuity respect to $s$ we have
            \begin{equation}\label{eq:thm-5.4-2}
            \begin{alignedat}{2}
                \norm{(t-s_0)^{-1}C(t,s_0)\Phi - (t-s)^{-1}C(t,s)\Phi}
                &\leq \sum_{i=1}^{n+1} \norm{g_i(t,s_0) S_i\tilde{S}(s_0)^{-1}\Phi - g_i(t,s) S_i \tilde{S}(s)^{-1} \Phi}  \\
                &\leq \sum_{i=1}^{n+1} |g_i(t,s_0)| \norm{S_i\tilde{S}(s_0)^{-1}\Phi - S_i \tilde{S}(s)^{-1} \Phi} + \\
                & \phantom{leq} + \sum_{i=1}^{n+1}|g_i(t,s_0) - g_i(t,s)| \norm{S_i\tilde{S}(s)^{-1}\Phi}.
            \end{alignedat}
            \end{equation}
            Let us examine separately the two terms on the right-hand side. First,
            \[
                \sum_{i=1}^{n+1}|g_i(t,s_0) - g_i(t,s)| \norm{S_i \tilde{S}(s)^{-1}\Phi}
                \leq K \sum_{i=1}^{n+1}|g_i(t,s_0) - g_i(t,s)| \norm{\Phi}
            \]
            and therefore because $g_i$ is uniformly continuous for $s\neq t$ in a compact subinterval of $I$, for every $s \neq t$ and every $\varepsilon > 0$ there exists $\delta_1 > 0$ such that for $|s_0 - s| < \delta_1$ it holds
            \[
                \sum_{i=1}^{n+1}|g_i(t,s_0) - g_i(t,s)| \norm{S_i \tilde{S}(s)^{-1}\Phi}
                \leq \frac{\varepsilon}{2}.
            \]

            For the first term in Eq. \eqref{eq:thm-5.4-2},
            \[
                \sum_{i=1}^{n+1} |g_i(t,s_0)| \norm{S_i \tilde{S}(s_0)^{-1}\Phi - S_i \tilde{S}(s)^{-1} \Phi} \leq K\sum_{i=1}^{n+1} |g_i(t,s_0)| \norm{\tilde{S}(s)\tilde{S}(s_0)^{-1}\Phi - \Phi}
            \]
            and thus by Lemma~\ref{lemma:2} and the fact that $g_i$ is uniformly bounded for $s \neq t$, for every $s \neq t$ and every $\varepsilon > 0$ there exists $\delta_2 > 0$ such that for $|s_0 - s| < \delta_2$ it holds
            \[
                \sum_{i=1}^{n+1} |g_i(t,s_0)| \norm{S_i \tilde{S}(s_0)^{-1}\Phi - S_i \tilde{S}(s)^{-1} \Phi} \leq \frac{\varepsilon}{2}.
            \]
            Hence, taking $\delta = \min\{\delta_1, \delta_2\}$ and substituting into Eq. \eqref{eq:thm-5.4-2} we have that for $|s_0 - s| < \delta$
            \[
                \norm{(t-s_0)^{-1}C(t,s_0)\Phi - (t-s)^{-1}C(t,s)\Phi} \leq \varepsilon
            \]
            which shows that hypothesis \ref{enum:reedSimon-ii} is fulfilled.

            Regarding hypothesis \ref{enum:reedSimon-iii} of Theorem~\ref{thm:reed-simon}, it is easy to see that $C(t) \Phi = \sum_{i=1}^{n+1} f_i'(t) S_i \tilde{S}(t)^{-1} \Phi$. Indeed, from Eq.~\ref{eq:thm-5.4-1} we get
            \[\begin{alignedat}{2}
                \norm{(t - s)^{-1} C(t,s) \Phi - \sum_{i=1}^{n+1} f_i'(t) S_i \tilde{S}(t)^{-1} \Phi}
                &= \norm{\sum_{i=1}^{n+1} \left[g_i(t,s) S_i \tilde{S}(s)^{-1} - f_i'(t) S_i \tilde{S}(t)^{-1}\right]\Phi} \\
                & \leq \sum_{i=1}^{n+1} |f_i'(t)|\norm{S_i\tilde{S}(s)^{-1}\Phi - S_i\tilde{S}(t)^{-1}\Phi} + \\
                &\phantom{\leq} + \sum_{i=1}^{n+1} |g_i(t,s) - f_i'(t)| \norm{S_i \tilde{S}(s)^{-1}\Phi}.
            \end{alignedat}\]
            Using again Lemma~\ref{lemma:1}, Lemma~\ref{lemma:2} and the fact that we are considering a compact subinterval, the continuity of every $f_i'$ and the definition of derivative implies the limit $C(t)\Phi = \lim_{s \to t} (t-s)^{-1} C(t,s)\Phi$ exists uniformly on $t$ and is equal to $C(t)\Phi$.

            Boundedness of $C(t)$ as an operator follows directly from Lemma~\ref{lemma:1} and the continuity of $f_i'(t)$:
            \[
                \norm{C(t) \Phi} = \norm{\sum_{i=1}^{n+1} f_i'(t) S_i \tilde{S}(t)^{-1} \Phi}
                \leq \sum_{i=1}^{n+1} |f_i'(t)| \norm{S_i \tilde{S}(t)^{-1} \Phi}
                \leq 2K \sum_{i=1}^{n+1} |f_i'(t)| \norm{\Phi}.
            \]
        \end{proof}

        Besides existence of unitary propagators for Schrödinger equations associated with Hamiltonians of the type we are dealing with, we are going to need a result on how close the evolution induced by two of these Hamiltonians is when they are similar (in the precise sense introduced in Theorem~\ref{thm:aprox-Hamiltonians-aprox-sol}).

        \begin{theorem}\label{thm:aprox-Hamiltonians-aprox-sol}
            Let $H_i$ be symmetric operators with common domain $\DD$. Let $f_i, g_i \in C^{1}(I)$, $i=1,\dots,n$ and $I\subset\R$. Suppose that $H_1(t) = \sum_{i = 1}^n f_i(t) H_i$ and $H_2(t) = \sum_{i = 1}^n g_i(t) H_i$, with common domain $\DD$ are self-adjoint operators, satisfying the hypothesis of Theorem~\ref{thm:timedependent-linearCombination-Hamiltonian}. Then, for every $\Psi \in \DD$, every $T > 0$ and every $\varepsilon > 0$ there exist $\delta> 0$ such that $\norm{f_i - g_i}_\infty < \delta$ implies $\norm{U_1(T, s) \Psi - U_2(T, s) \Psi} < \varepsilon$.
        \end{theorem}
        \begin{proof}
            By Theorem~\ref{thm:timedependent-linearCombination-Hamiltonian}, there exist unitary propagators $U_1(t, s)$, $U_2(t, s)$ associated with $H_1(t)$, $H_2(t)$ respectively. Since for any $\Psi \in \DD$, $t \mapsto U_\l(t, s) \Psi$ is strongly differentiable ($\l = 1, 2$), we have $t \mapsto \norm{U_1 (t, s) \Psi - U_2 (t, s) \Psi}$ is differentiable and by the Fundamental Theorem of Calculus we have
            \[
                \norm{U_1 (T, s) \Psi - U_2 (T, s) \Psi}
                = \int_s^T \Dt \norm{U_1 (t, s) \Psi - U_2 (t, s) \Psi} dt.
            \]
            Strong differentiability implies that we can take the derivative into the norm and get
            \[ \begin{alignedat}{2}
                \norm{U_1 (T, s) \Psi - U_2 (T, s) \Psi}
                &= \int_s^T \norm{\Dt U_1 (t, s) \Psi - \Dt U_2 (t, s) \Psi} dt \\
                &= \int_s^T \norm{H_1(t) \Psi - H_2(t) \Psi} dt \\
                &\leq \sum_{i=1}^n \int_s^T \abs{f_i - g_i} \norm{H_i \Psi} dt \\
                & \leq (T - s) \sum_{i=1}^n \norm{f_i - g_i}_\infty \norm{H_i \Psi}
            \end{alignedat} \]
            Hence, it is enough to take
            \[ \delta =\inf_i \frac{\varepsilon}{n (T - s) \norm{H_i \Psi}}. \]
        \end{proof}


\section{Approximate Controllability of Boundary Control Systems}\label{sec:approximatecontrol}

    Approximate controllability of the boundary control system is, by Proposition~\ref{prop:equivmagneticLaplacian}, equivalent to approximate controllability of a quantum system with Hamiltonian

            \[ H(t) = -\left[\left(\Dx - i A(t)\right)^2 + A'(t) x \right] \]
            and periodic boundary conditions.

            Unlike control on a finite dimensional Hilbert space, control on an infinite dimensional Hilbert space has no general result giving necessary and sufficient conditions for (approximate) controllability. However, it will be enough for us to rely on a theorem by T. Chambrion et al.\ \cite{chambrion_controllability_2009}, giving sufficient conditions to prove approximate controllability for the boundary control system. In the referenced work it is studied the approximate controllability of some linear control systems; that is, systems whose evolution is given by
            \begin{equation}\label{eq:bilinear-schrodinger-control}
                i\Dt \Psi(t) = H_0 \Psi(t) + u(t) H_1 \Psi(t),
            \end{equation}
            with $u: \R \to (0, c)$. Moreover, they assume that:
            \begin{enumerate}[%
                    label=\textit{(A\arabic*)},
                    nosep,
                    labelindent=0.5\parindent,
                    leftmargin=*
                    ]
                \item\label{H1} $H_0, H_1$ are self-adjoint operators not depending on $t$,
                \item\label{H2} there exists an orthonormal basis $\{\phi_n\}_{n \in \N}$ of $\H$ made of eigenvectors of $H_0$, and
                \item\label{H3} $\phi_n \in \dom H_1$ for every $n \in \N$.
            \end{enumerate}
            Control systems satisfying conditions \ref{H1}--\ref{H3} will be called \textbf{normal quantum control systems}. For them, the following theorem is proven:
            \begin{theorem}[Chambrion et al.\ {\cite[Thm. 2.4]{chambrion_controllability_2009}}.] \label{thm:chambrion-controllability}
                Consider a normal quantum control system, with $c > 0$ as described above. Let $\{\lambda_n\}_{n \in \N}$ denote the eigenvalues of $H_0$, each of them associated to the eigenfunction $\phi_n$. Then, if the elements of the sequence $\{\lambda_{n+1} - \lambda_n\}_{n \in \N}$ are $\Q$-linearly independent and if $\langle \phi_{n+1}, H_1 \phi_n \rangle \neq 0$ for every $n \in \N$, the system is approximately controllable.
            \end{theorem}

            Based on this theorem we will prove the main result of this work which ensures the approximate controllability of the boundary control systems. This is the first instance in which controllability of a system using boundary controls is considered. Before doing that it is convenient to introduce the following lemma:
            \begin{lemma}\label{lemma:chambrion-thm-applies}
                Consider a normal quantum control system $i \Dt \Psi = H_0 \Psi + u(t) H_1 \Psi$ with $H_0$, $H_1$ such that $H(t) = H_0 + u(t)H_1$ satisfies hypothesis of Theorem~\ref{thm:timedependent-linearCombination-Hamiltonian}. Then given any $\varepsilon > 0$ there exist perturbed Hamiltonians $\tilde{H}_0$, $\tilde{H}_1$ with the same domain as $H_0$ such that they satisfy the conditions of Theorem~\ref{thm:timedependent-linearCombination-Hamiltonian} and also those of Theorem~\ref{thm:chambrion-controllability} and such that for every $t > s$ and every $C^1$ piecewise function $u:[s,t] \to \R$, it holds:
                \[
                    \norm{U(t, s) \Psi - \tilde{U}(t,s) \Psi} < \varepsilon, \qquad
                    (\forall \Psi \in \dom H_0),
                \]
                where we denote by $U(t,s)$ and $\tilde{U}(t,s)$ the unitary propagators associated with $H(t)$ and $\tilde{H}(t) = \tilde{H}_0 + u(t) \tilde{H}_1$ respectively.
            \end{lemma}
            \begin{proof}
                Let $\lambda_k, \phi_k$ denote the eigenpairs of $H_0$. If it is the case that $\Q$-linear independence condition of Theorem~\ref{thm:chambrion-controllability} hold, we set $\tilde{H}_0 = H_0$; otherwise take an increasing sequence of positive irrational numbers $\nu_k$ such that they are rationally independent and $\nu_k < 2^{-k}$. Then define
                \[ H_{0,p} = \sum_{k \in \N} \nu_k \phi_k \phi_k^\dagger. \]
                Obviously $H_{0,p}$ has the same domain as $H_0$ because $\norm{H_{0,p}\Psi}^2 \leq \sum_k 2^{-2k}  |\langle \phi_k, \Psi \rangle |^2$, from what we get $\norm{H_{0,p}\Psi} \leq \norm{\Psi}$ and $H_{0,p}$ can be chosen to have the same domain as $H_0$.

                Define $\tilde{H}_0 = H_0 + \mu_0 H_{0,p}$ with $\mu_0 \in \Q$. Then $\tilde{H}_0$ has eigenvalues $\lambda_k + \mu_0 \nu_k$ satisfying the rationally independence condition of Theorem~\ref{thm:chambrion-controllability}.

                If $H_1$ is such that $\langle \phi_{n+1}, H_1 \phi_n \rangle = 0$ for $n \in N \subset \N$, take a sequence of positive, non-vanishing terms $\{\alpha_n\}_{n \in N}$ such that $\alpha_n < 2^{-n}$ and define
                \[ H_{1,p} = \sum_{n \in N} \alpha_n \phi_{n+1} \phi_n^\dagger. \]
                Again the domain of $H_{1,p}$ can be chosen to be $\dom H_0$, since $\norm{H_{1,p} \Psi}^2 \leq \sum_{n \in N} 2^{-2n} |\langle \phi_n, \Psi \rangle|^2$ and thus $\norm{H_{1,p}\Psi} \leq \norm{\Psi}$.

                Defining $\tilde{H}_1 = H_1 + \mu_1 H_{1,p}$ with $\mu_1$ real it is clear that $\tilde{H}_1$ satisfies $\langle \phi_{n+1}, \tilde{H}_1 \phi_n \rangle \neq 0$ for any $n \in \N$.

                From what we already said, taking into account that $H_{0,p}$ and $H_{1,p}$ are bounded, it follows that if $H(t)$ satisfies the hypothesis of Theorem~\ref{thm:aprox-Hamiltonians-aprox-sol}, so does $\tilde{H}(t)$ on each interval in which $u(\tau)$ is $C^1$ and therefore, taking $\mu_0$ and $\mu_1$ small enough we have
                \[
                    \norm{U(t, s) \Psi - \tilde{U}(t,s) \Psi} < \varepsilon, \qquad
                    \text{for all } \Psi \in \dom H_0.
                \]
            \end{proof}

            Suppose that $A\in C^2(I) $ defines the time dependent magnetic vector potential. Then $H(t)$ fulfills all the hypothesis of Theorem~\ref{thm:timedependent-linearCombination-Hamiltonian} but \ref{enum:H3-thm-linCombHamiltonian}, which requires some work to prove. The following lemma shows that the families of Hamiltonians $H(t)$ that we consider, i.e., those on Proposition~\ref{prop:equivmagneticLaplacian}, satisfy hypothesis $\ref{enum:H3-thm-linCombHamiltonian}$ of Theorem~\ref{thm:timedependent-linearCombination-Hamiltonian} and therefore have well defined evolutions.

            \begin{lemma}\label{lemma:magnetic-Laplacian-satisfy-thm5.3}
                Let $A: L \to \mathbb{R}$ be a constant magnetic vector potential, and denote by $-D_A^2$ a self-adjoint extension of the associated magnetic Laplacian, whose domain we denote by $\mathcal{D} \subset \mathcal{H}^2(L)$. Let $r > 0$, and suppose that $ |A| < r$. Then there exists a constant $K$ (not depending on $A$) such that
                \[
                    \norm{\D[2]{\Psi}{x}} \leq K \left(\norm{D_A^2 \Psi} + \norm{\Psi}\right)
                \]
                for all $\Psi \in \mathcal{D}$.
            \end{lemma}
            \begin{proof}
                We will prove this lemma in three steps. First, we will show that for every constant vector potential $A\in \H^1(L)$ the bound in the lemma stands with constant $K_A$ depending on $A$. In fact, since the magnetic Laplacian is self-adjoint, $-D_A^2 + iI$ is invertible with bounded inverse and maps $\H$ onto $\mathcal{D}\subset\H^2(L)$. Also $-\Dx[2]$ is closed in $\mathcal{H}^2(L)$ since it is the adjoint of the standard Laplacian with the minimal symmetric domain ($i.e.$, with domain $\mathcal{D}_0 = \{\Psi \in \mathcal{H}^2(L)\mid  \veval{\Psi} = 0, \veval{\frac{\d\Psi}{\d x}} = 0\}$) as is well-known (see, e.g., \cite{lions_problemes_1968}). By the Closed Graph Theorem this implies that
                $$\Dx[2] (D_A^2 + iI)^{-1}$$
                 is a bounded operator on $\mathcal{L}^2(L)$. Therefore, for any $\Psi \in \mathcal{D}$,
                \begin{equation}\label{eq:magneticLaplacian-lowerBound-1}
                    \norm{\D[2]{\Psi}{x}} = \norm{\Dx[2] (D_A^2 + iI)^{-1} (D_A^2 + iI) \Psi}
                    \leq K_A \norm{(D_A^2 + iI) \Psi} \leq K_A \left(\norm{D_A^2 \Psi} + \norm{\Psi}\right)
                \end{equation}
                where we have defined $K_A = \norm{\Dx[2] (D_A^2 + iI)^{-1}}$.

                Once we have proved Equation \eqref{eq:magneticLaplacian-lowerBound-1}, we can prove that for every constant vector potential $A$ there exists an $\varepsilon_A > 0$ such that for any constant magnetic vector potential $B$ satisfying $|A  - B | \leq \varepsilon_A$ it holds
                \[
                    \norm{\D[2]{\Psi}{x}} \leq \tilde{K}_A \left(\norm{D_B^2 \Psi}+\norm{\Psi}\right),
                \]
                with $\tilde{K}_A > 0$ not depending on $B$. Indeed, from Equation \eqref{eq:magneticLaplacian-lowerBound-1} we have
                \begin{equation}\label{eq:magneticLaplacian-lowerBound-2}
                    \norm{\D[2]{\Psi}{x}} \leq K_A \left(\norm{D_A^2 \Psi} + \norm{\Psi}\right)
                    \leq K_A \left(\norm{D_A^2 \Psi - D_B^2 \Psi} + \norm{D_B^2\Psi} + \norm{\Psi} \right).
                \end{equation}
                Let us examine the first term in the parenthesis. By the definition of the Magnetic Laplacians,
                \begin{equation}\label{eq:magneticLaplacian-lowerBound-3}
                        \norm{D_A^2 \Psi - D_B^2 \Psi}
                        = \norm{(A ^2 - B ^2) \Psi  + 2i(A  - B )\D{\Psi }{x }}.
                \end{equation}
                Denoting $\varepsilon  = |A  - B |$ and using the triangle inequality one gets
                \begin{equation}
                    \norm{(A ^2 - B ^2) \Psi  + 2i(A  - B )\D{\Psi }{x }}
                    \leq \varepsilon (2|A | + \varepsilon ) \norm{\Psi } + 2 \varepsilon  \norm{\D{\Psi }{x }}.
                \end{equation}

                Now, using the well-known fact \cite[Thm. 5.2]{adams_sobolev_2003} that
                \[
                    \norm{\D{\Psi }{x }} \leq \tilde{K} \left(\norm{\D[2]{\Psi }{x }} + \norm{\Psi }\right),
                \]
                we get
                \begin{equation}
                    \norm{(A ^2 - B ^2) \Psi  + 2i(A  - B )\D{\Psi }{x }} \leq \varepsilon(2|A | + \varepsilon + 2 \tilde{K}) \left( \norm{\Psi} + \norm{\D[2]{\Psi}{x}}\right).
                \end{equation}

                Let us define the function $\kappa(\varepsilon) := \varepsilon(2 |A | + \varepsilon + 2 \tilde{K})$, which is a continuous monotone function of $\varepsilon$ with range $[0, \infty)$. Substituting back into into Equation~\ref{eq:magneticLaplacian-lowerBound-3} we get
                \[
                    \norm{D_A^2 \Psi - D_B^2 \Psi} \leq \kappa(\varepsilon) \left(
                        \norm{\Psi} + \norm{\D[2]{\Psi}{x}}
                    \right).
                \]

                Hence, from Equation \eqref{eq:magneticLaplacian-lowerBound-2} we have that
                \[
                    \left(1 - \kappa(\varepsilon) K_A\right) \norm{\D[2]{\Psi}{x}}
                    \leq K_A\left(1 + \kappa(\varepsilon) \right) \left(\norm{D_B^2\Psi} + \norm{\Psi}\right).
                \]
                Obviously we can choose $\varepsilon_A$ such that $\kappa(\varepsilon_A) K_A  = 1/2$, and then
                \[
                    \norm{\D[2]{\Psi}{x}} \leq 2K_A\left(1 + \kappa(\varepsilon_A) \right)
                    \left(\norm{D_B^2\Psi} + \norm{\Psi}\right)
                    \eqqcolon \tilde{K}_A \left(\norm{D_B^2\Psi} + \norm{\Psi}\right).
                \]

                The proof can be finished by a compacity argument. Since we are only considering constant vector potentials, each potential $A$ defines a point in $\mathbb{R}$. The subset $\mathcal{K}$ of $\mathbb{R}$ associated to the set of vector potentials satisfying $| A | \leq r$ is a compact subset. Now, define $U_A = \{B \in \mathbb{R} \mid  |B  - A | < \varepsilon_A\}$; the family $\{U_A\}_{A \in U}$ forms a covering of $\mathcal{K}$ and by compacity it admits a finite subcovering $\{U_{A_i}\}_i$. Taking the maximum of the associated constants,
                \[
                    K = \max_i \tilde{K}_{A_i},
                \]
                concludes the proof.
            \end{proof}

            \begin{theorem}\label{thm:controllability-piecewise-smooth}

            Let $\mathcal{C}_p(0,T)$ the set of piecewise two times continuously differentiable functions on the interval $[0,T]$. The boundary control system with controls $\mathcal{C}_p(0,T)$ is approximately controllable.
            \end{theorem}
            \begin{proof}
                By Proposition~\ref{prop:equivmagneticLaplacian} the boundary control system is controllable if and only if so is the magnetic controlled system given by
                \begin{equation}\label{eq:original-hamiltonian}
                    H(t) = - \left[\left(\Dx - iA(t)\right)^2 +  A'(t) x \right]
                \end{equation}
                and periodic boundary conditions. We will proof that this equivalent system is approximately controllable using Theorem~\ref{thm:chambrion-controllability}. The main problem to do this is the fact that $A(t)$ and $A'(t)$ are not independent, and to avoid this problem we need to proceed in two steps. First we define an auxiliary system to which Theorem~\ref{thm:chambrion-controllability} applies and then we use Theorem~\ref{thm:aprox-Hamiltonians-aprox-sol} to show that for any controls on the auxiliary system, its evolution is approximately the same as the evolution of the original system with some controls related to those on the auxiliary system.

                Let us start with the first step. Take a constant magnetic vector potential $a \in  \H^1(L)$ with associated magnetic Laplacian $-D^2 = -\left(\Dx - ia\right)^2$. We consider
                \begin{equation}\label{eq:A-prime}
                    A'(t) = u(t),
                \end{equation}
                where $u: I \to \R$ is a control, and define the auxiliary system with Hamiltonian
                \begin{equation}\label{eq:auxiliary-hamiltonian}
                    \tilde{H}(t) = -D^2 - u(t)  x,
                \end{equation}

                It is easy to check that the assumptions made by Chambrion et al.\ are satisfied in our case: $H_0 = -D^2$ and $H_1 =  x$ are self-adjoint operators not depending on $t$, there exists an orthonormal basis of the Hilbert space $\H$ made of eigenfunctions of any magnetic Laplacian over $L$ provided that $L$ is compact \cite[Thm. 3.1.1]{berkolaiko2013introduction}, and $H_1$ is a self-adjoint bounded operator (since $L$ is compact) and thus $\dom H_1 = \H$.

                By Lemma~\ref{lemma:chambrion-thm-applies}, Theorem~\ref{thm:chambrion-controllability} can be applied (either to $\tilde{H}(t)$ or to a perturbed system with evolution as \emph{closed} as desired) and so the system is approximately controllable. Hence, for every initial state $\Psi_0$, every target state $\Psi_T$, every $\varepsilon > 0$ and every $c > 0$ there exists $T>0$ and $u(t): [0, T] \to (0, c)$ piecewise constant such that the evolution induced by $\tilde{H}(t)$ and denoted $\tilde{\Psi}(t)$, satisfies $\tilde{\Psi}(0) = \Psi_0$ and $\norm{\tilde{\Psi}(T) - \Psi_T} < \varepsilon / 2$. Denote by $\tilde{U}(t,s)$ the unitary propagator associated to $\tilde{H}$ with controls $u(t)$.

                Now, choosing the vector potential from the original system \eqref{eq:original-hamiltonian} in such a way that its induced evolution is close enough to that of the auxiliary system, one guarantees that the evolved state reaches near the target state at time $T$. In order to do that, we split the time interval $[0, T]$ into $N$ pieces of length $\tau = T/N$, and for each of those subintervals define $A_k: [k\tau, (k+1)\tau) \to \mathbb{R}$ as
                \[
                    A_k(t) = a + \int_{k\tau}^t u(s) \, ds,
                \]
                with $u(t)$ the piecewise control given by Chambrion et al.'s theorem. Taking
                \[
                    A(t) = \sum_{k = 0}^{N - 1} \chi_{[k \tau, (k+1) \tau)}(t) A_k(t),
                \]
                it is clear that $A'(t) = u(t)$ and $A \in \mathcal{C}_p(0,T)$. Also, by the mean value theorem,
                \begin{equation}\label{eq:controllability-3}
                    \norm{A - a}_\infty =
                    \adjustlimits\max_{k < N} \sup_{k\tau \leq t < (k+1)\tau}
                    \int_{k\tau}^t u(s) \,ds \leq c \tau
                \end{equation}
                For the moment, $\tau$ is arbitrary but later on we will need to choose it small enough.

                Expanding the square on \eqref{eq:original-hamiltonian} and having into account Lemma~\ref{lemma:magnetic-Laplacian-satisfy-thm5.3}, it is easy to check that the Hamiltonian of the original system, $H(t)$, fulfills the hypothesis of Theorem~\ref{thm:timedependent-linearCombination-Hamiltonian} in every interval $[k\tau, (k+1)\tau)$.

                Hence, there exists a unitary propagator $U_k(t,s)$ describing the evolution induced by it for $t, s \in [k\tau, (k+1)\tau)$. For $t \in [k\tau, (k+1)\tau]$, $s \in [\ell \tau,
                (\ell + 1)\tau)$ with $\ell < k$ the unitary propagator is constructed multiplying them:
                \[
                    U(t, s) = U_k(t, k\tau) U_{k-1}(k\tau, (k-1)\tau) \cdots U_\ell((\ell + 1)\tau, s).
                \]
                In what follows we omit the subscript on $U_k$ since the values of its arguments $t,s$ identify the index $k$ unambiguously.

                Finally, let $\{I_j\}$ with $I_j = [t_j, t_{j -1})$ be the coarser partition of $[0, T]$ which is a common refinement of both the partition $\{[k\tau, (k+1)\tau)\}_{k}$ and that given by the piecewise definition of $u(t)$. Remember that Theorem~\ref{thm:chambrion-controllability} proves approximate controllability for $u(t)$ piecewise-constant control functions. It is clear that the state of the system at time $T \in I_n$, assumed the evolution induced by $H(t)$ (defined in Equation \eqref{eq:controllability-4}) starting at $\Psi_0$, can be written as
                \[
                    \Psi(T) = U(T, t_n) U(t_n, t_{n-1}) \cdots U(t_1, 0) \Psi_0.
                \]
                And similarly for the state $\tilde{\Psi}(T)$ if we assume evolution by $\tilde{H}$ defined in Equation \eqref{eq:auxiliary-hamiltonian} (using the unitary propagator $\tilde{U}$ instead of $U$).

                It is straightforward to check that in every $I_j$ both Hamiltonians satisfy the hypothesis of Theorem~\ref{thm:aprox-Hamiltonians-aprox-sol}: the domain of magnetic Laplacians is fixed by periodic boundary conditions independent of $t$, and the multiplication operator $x$ is bounded. Remember that $u(t)$ being $C^1(I_j)$, in fact constant on $I_j$, implies that the functions giving the time dependence of the Hamiltonians (after expanding the magnetic Laplacians) are also $C^1(I_j)$. Both Hamiltonians satisfy the hypothesis of Theorem~\ref{thm:timedependent-linearCombination-Hamiltonian} (see Lemma~\ref{lemma:magnetic-Laplacian-satisfy-thm5.3}). Hence, for any $t, s \in I_j,$ and any $\varepsilon_2 > 0$ we can chose $\delta = c \tau$ as in Theorem~\ref{thm:aprox-Hamiltonians-aprox-sol} so that
                \[ \norm{\tilde{U}(t, s) \Psi_0 - U(t, s) \Psi_0} < \varepsilon_2. \]
                Hence, we have
                \[ \begin{alignedat}{2}
                    \norm{\tilde{\Psi}(T) - \Psi(T)}
                    &= \norm{\tilde{U}(T, t_n) \cdots \tilde{U}(t_1, s) \Psi_0 -
                             U(T, t_n) \cdots U(t_1, s) \Psi_0} \\
                    &\leq \norm{\tilde{U}(T, t_n) \cdots \tilde{U}(t_2, t_1)U(t_1, s) \Psi_0 - U(T, t_n) \cdots U(t_1, s) \Psi_0} + \varepsilon_2 \\
                    &\vdots \\
                    &\leq (n+1) \varepsilon_2.
                \end{alignedat} \]
                Taking $\varepsilon_2 = \varepsilon / (2n+2)$, we have
                \[ \norm{\tilde{\Psi}(T) - \Psi(T)} \leq \frac{\varepsilon}{2}. \]

                Using that for the auxiliary system we have that $\norm{\tilde{\Psi}(T) - \Psi_T} < \frac{\varepsilon}{2}$, we conclude
                \[ \norm{\Psi(T) - \Psi_T} < \varepsilon. \]
                Hence, we have found a control $A:[0,T]\to \R$ piecewise two times continuously differentiable such that from any $\Psi_0$ we can reach as close as we want to any $\Psi_T$ and so the system is approximately controllable.
            \end{proof}
            Using Theorem~\ref{thm:controllability-piecewise-smooth} and an approximating argument similar to that in its proof, is easy to show that controls can also be smooth functions of time.
            \begin{corollary}\label{corol:controllability-smooth}
                Every boundary control system with smooth controls $A:[0,T] \to \R$ is approximately controllable.
            \end{corollary}
            \begin{proof}
                By Proposition~\ref{prop:equivmagneticLaplacian} the boundary control system is approximately controllable if and only if so is the magnetic controlled system with
                \[ H(t) = - \left[\left(\Dx - iA(t)\right)^2 + A'(t) x \right] \]
                and periodic boundary conditions.

                From Theorem~\ref{thm:controllability-piecewise-smooth}, for every initial state $\Psi_0$, every target state $\Psi_T$ and every $\varepsilon > 0$, we have piecewise two times continuously differentiable controls $\tilde{A}(t)$ such that the evolution $\tilde{\Psi}(t)$ induced by
                \[ \tilde{H}(t) = - \left[\left(\Dx - i\tilde{A}(t)\right)^2  + \tilde{A}'(t) x \right], \]
                satisfies $\tilde{\Psi}(0) = \Psi_0$ and $\norm{\smash{\tilde{\Psi}(T) - \Psi_T}} < \varepsilon / 2$. Denote by $\tilde{U}(t,s)$ the unitary propagator associated to $\tilde{H}(t)$.

                Using some well-known approximation result (see, for example, \cite[\S5.3]{evans_partial_1998}) one can find $A(t)$ smooth such that $\norm{A - \tilde{A}} < \delta_1$ and $\norm{A' - \tilde{A}'} < \delta_2$. Taking $\{I_j\}_j$ the partition of $[0, T]$ given by the subintervals on which $\tilde{A}(t)$ is $C^2$ and using the same argument as in the proof of Theorem~\ref{thm:controllability-piecewise-smooth}, one can use Theorem~\ref{thm:aprox-Hamiltonians-aprox-sol} to show that the evolution induced by $H(t)$ satisfies
                \[ \norm{\Psi(T) - \Psi_T} < \varepsilon.\]
            \end{proof}


\section{Conclusions}

We proposed a scheme for quantum control at the boundary and rigorously proved its controllability. It is the first time that the controllability of such a quantum system has been considered. This shows that the scheme of quantum control at the boundary is feasible. Moreover, the particular system considered presents the advantage that it could be experimentally implemented. Indeed, this quantum system represents a quantum particle moving in a spire controlled by the flux of a magnetic field that traverses the plane of the spire.

\subsubsection*{Acknowledgments.}
J.M.P.P. is supported by QUITEMAD+, S2013/ICE-2801 and the “Juan de la Cierva - Incorporación” Proyect 2018/00002/001 and is partially supported by Spanish Ministry of Economy and Competitiveness through project DGI \hbox{MTM2017-84098-P.}

A.B. is supported by QUITEMAD+, S2013/ICE-2801 and the UC3M University through Ph.D. program grant PIPF UC3M 01-1819 and is partially supported by Spanish Ministry of Economy and Competitiveness through project DGI MTM2017-84098-P.


\end{document}